\documentclass[leqno]{amsart}
\usepackage{amsfonts,amssymb,amsmath,amsgen,amsthm}
\usepackage[notref,notcite]{}%{showkeys}
\usepackage{hyperref,color}
\usepackage{pdfsync}
\usepackage{bm}
\usepackage{bbm}
\usepackage{latexsym}

\usepackage{mathrsfs}
\usepackage{amsthm}
\usepackage{amsfonts}
\usepackage{amssymb}
\usepackage{amscd}

\usepackage{color}

\theoremstyle{plain}
\newtheorem{theorem}{Theorem}[section]
\newtheorem{definition}[theorem]{Definition}
\newtheorem{lemma}[theorem]{Lemma}

\theoremstyle{remark}
\newtheorem{remark}[theorem]{Remark}

\def\C{{\mathbb C}}% complex numbers
\def\R{{\mathbb R}}% real numbers
\def\({\left(}
\def\){\right)}
\def\<{\left\langle}
\def\>{\right\rangle}

\def\1{{\mathbf 1}}

\def\eps{\varepsilon}

\DeclareMathOperator{\RE}{Re}

\DeclareMathOperator{\supp}{supp}

\numberwithin{equation}{section}

\date\today

\title[Point interactions on bounded domains]{Approximation of Schr\"odinger operators with point interactions on bounded domains}

\author[Diego Noja]{Diego Noja}
\address{Dipartimento di Matematica e Applicazioni, Universit\`a di Milano Bicocca, Via~R.~Cozzi 55, 20126 Milano.}
\email{diego.noja@unimib.it}

\author[Raffaele Scandone]{Raffaele Scandone}
\address{Universit\`a degli Studi di Napoli ``Federico II'', Dipartimento di Matematica e Applicazioni ``R.~Caccioppoli'', Complesso Monte S.~Angelo - Via Cintia, 80126 Napoli.}
\email{raffaele.scandone@unina.it}

\begin{document}
\begin{abstract}
We consider Schr\"odinger operators on a bounded domain $\Omega\subset \R^3$, with homogeneous Robin or Dirichlet boundary conditions on $\partial\Omega$ and a point (zero-range) interaction placed at an interior point of $\Omega$. We show that, under suitable spectral assumptions, and by means of an extension-restriction procedure which exploits the already known result on the entire space, the singular interaction is approximated by rescaled sequences of regular potentials. The result is missing in the literature, and we also take the opportunity to point out some general issues in the approximation of point interactions and the role of zero energy resonances.
\end{abstract}

\maketitle

\section{Introduction and main result}
In this paper we consider a class of Schr\"odinger operators with point interactions on bounded domains, and we study their approximation through regular potentials.
 Let $\Omega$ be an open, connected, bounded domain in $\R^3$, with boundary $\partial\Omega$ of class $\mathcal{C}^2$. We consider the Laplace operator $-\Delta_{\sigma}$ on $L^2(\Omega)$, with the homogeneous Robin boundary condition
 \begin{equation}\label{robin}
\sigma(\psi)(x):=\partial_{n}\psi(x)+b(x)\psi(x)=0,\qquad  x\in\partial\Omega,
\end{equation} 
%\begin{equation}\label{robin}
%\sigma(\psi):=a(x)\partial_{n}\psi(x)+b(x)\psi(x)=0,\qquad  x\in\partial\Omega,
%\end{equation}
or the homogeneous Dirichlet boundary condition
 \begin{equation}\label{D}
\psi(x)=0,\qquad  x\in\partial\Omega,
\end{equation} 
formally corresponding to the case $b(x)=+\infty$.
In \eqref{robin} $\partial_n$ denotes the outward normal derivative with respect to $\partial\Omega$, and $b$ is a real valued, $\mathcal{C}^1$-function on $\partial\Omega$. The case $b=0$ yields the Neumann boundary condition. %satisfying a nondegeneracy condition (see Theorem \ref{main}){\color{blue} (condizione di non degenerazione, credo basti che $a$ e $b$ non si annullino mai simultaneamente)}. Observe that \eqref{robin} includes in particular the classical Dirichlet and Neumann boundary conditions.

In addition, we perturb $-\Delta_{\sigma}$ through a delta-like interaction at a point $x_0\in\Omega$, which formally amounts to introduce in the operator domain singular elements of the form 
\begin{equation}\label{eq:BPcontact}
	g(x)\;=\;\mathrm{const}\cdot\big(\,|x-x_0|^{-1}+4\pi\alpha\big)+o(1)\qquad\mathrm{as}\;\;x\to x_0\,,
\end{equation}
with a certain fixed $\alpha\in\R$. The relation between the singular part and the regular part of the domain element displayed in the previous formula is a boundary condition at the singularity characterizing point interactions in three dimensions.
In the context of Quantum Mechanics, \eqref{eq:BPcontact} is the short-range asymptotic approximation of the low-energy bound state of a potential ideally supported at $x_0$ and with $s$-wave scattering length $-(4\pi\alpha)^{-1}$, see e.g.~the discussion in the seminal paper by Bethe and Peierls \cite{Bethe_Peierls-1935}, and the treatise \cite{albeverio-solvable} for a rigorous analysis. From now on we suppose, without loss of generality, $x_0=0\in\Omega$.

Given $\alpha\in\R$, we denote by $-\Delta_{\alpha,\sigma}$, the self-adjoint operator on $L^2(\Omega)$ acting as the (suitably re-normalized) Laplacian on functions satisfying \eqref{robin} and the singular boundary condition \eqref{eq:BPcontact}. The case $\alpha=+\infty$ corresponds to the unperturbed Laplace operator $-\Delta_{\sigma}$. We refer to Section \ref{se2} for the rigorous definition and the main properties of $-\Delta_{\alpha,\sigma}$.

Our goal is to realize $-\Delta_{\alpha,\sigma}$ as a suitable limit of regular, re-scaled Schr\"odinger operators on $L^2(\Omega)$. The analogous approximation problem on $\R^n$, $n=1,2,3,$ (the dimensions where a non-trivial point interaction does actually exist) is well-understood, see e.g.~\cite{albeverio-solvable,BHS2013,BHL2014,Mic-Scan,slk,GY,DMV24} and references therein. Periodic and anti-periodic boundary conditions have been considered in \cite{fassari,holdenAIHP}. The case of a 3d Riemannian manifold without boundary has been investigated in \cite{GIP}, but when the ambient space is a bounded domain and in the presence of general boundary conditions there are no explicit results in the literature.

Let us introduce our framework. Let us fix a real valued, compactly supported potential $V\in L^{\frac32}(\R^3)$ (the critical integrability with respect to the Laplace operator in 3d, as the scaling $V(x)\mapsto \tau^2V(\tau x)$, $\tau\in\R$, leaves the $L^{3/2}(\R^3)$ norm  invariant), and a parameter $\lambda\in\R$. For every $\eps>0$, we set
\begin{equation}\label{eq:Veps}
	V_{\eps}(x):=\frac{1+\lambda\eps}{\eps^2}V\left(\frac{x}{\eps}\right).
	\end{equation}
Under these assumptions, $-\Delta_{\sigma}+V_{\eps}$ defines a self-adjoint operator on $L^2(\Omega)$. We aim to study the limit of $-\Delta_{\sigma}+V_{\eps}$ as $\eps\to 0$ (our analysis also applies if the term $1+\lambda\varepsilon$ in \eqref{eq:Veps} is replaced by an analytic function $f=f(\varepsilon)$ with $f(0)=1$ and $f'(0)=\lambda$). We assume that there exists an open ball $B$ such that $\supp(V)\subseteq \overline{B}\subseteq\Omega$ (this is not restrictive, since $\Omega$ is open and for $\eps\to 0$ the support of $V_{\eps}$ shrinks to the origin). The case when the potential $V$ is not compactly supported (and then $V_{\varepsilon}$ does not vanish on the boundary for any $\varepsilon>0$) is discussed in the recent preprint \cite{LoPo}. We are going to show that, under suitable spectral assumptions, the presence of the re-scaled potential $V_{\eps}$ induces in the limit as $\eps\to 0$ a non-trivial point interaction.
The precise statement and proof needs the introduction of the Birman-Schwinger operator associated to $V$. 

To this aim, let us set \begin{equation}\label{uv}
u:=\operatorname{sgn}(V)|V|^{1/2}\ ,\ \ \ \ \ v:=|V|^{1/2},
\end{equation} and define
\begin{equation}\label{defB}
	B_0(f)(x):=\frac{u(x)}{4\pi}\int_{\Omega}\frac{v(y)f(y)}{|x-y|}dy.
	\end{equation}
A direct application of the Hardy-Littlewood-Sobolev and H\"older inequalities shows that $B_0$ is an Hilbert-Schmidt (whence compact) operator on $L^2(\Omega)$.\\
Our main result is the following.
\begin{theorem}\label{th:main}
Suppose that $-1$ is a simple eigenvalue for $B_0$, and let $\varphi$ be the associated normalized eigenfunction. Assume moreover that $\langle v,\varphi\rangle\neq 0$, and set $\alpha:=-\lambda|\langle v,\varphi\rangle|^{-2}\neq +\infty$. Then $-\Delta_{\sigma}+V_{\eps}$ converges to $-\Delta_{\alpha,\sigma}$ as $\eps\to 0$, in the norm resolvent sense in $L^2(\Omega)$.
\end{theorem}
\begin{remark}
We point out that the spectral assumption on the potential $V$ is insensitive on the Dirichlet or Robin boundary condition on $\partial\Omega$.
\end{remark}

\begin{remark}
When $-1$ is not an eigenvalue of $B_0$, $-\Delta_{\sigma}+V_{\eps}$ converges to the free Laplace operator $-\Delta_{\sigma}$, that is $-\Delta_{\alpha,\sigma}$ with $\alpha=+\infty$. The same happens when $-1$ is a simple eigenvalue but the corresponding eigenfunction $\varphi$ is orthogonal to $V$. The non-orthogonality condition $\langle v,\varphi\rangle\neq 0$ guarantees instead that in the limit as $\eps\to 0$ we obtain a non-trivial point interaction at the origin, i.e.~an operator not coinciding with the free Laplacian $\Delta_\sigma$.
\end{remark}

\begin{remark}
Theorem \ref{th:main} could be extended, with minor modifications, to the situation when $-1$ is a degenerate eigenvalue for $B_0$, in which case the limit as $\eps\to 0$ produces a non-trivial point interaction if and only if at least one of the eigenfunctions satisfies the non-orthogonality condition $\langle\varphi,v\rangle\neq 0$. 
\end{remark}

\begin{remark}
As is well known, the analogous result holds on $\R^3$ (see Theorem \ref{conv_tre} below). The proof on $\R^3$ relies upon a precise spectral analysis of the resolvent expansion of $-\Delta+V_{\eps}$, together with a rescaling of the space variables, an argument which cannot be directly adapted to the case of a bounded domain. Here we exploit instead a suitable extension-restriction approach, which allows to transfer, after some adaptation, the results on $\R^3$ to our framework.
\end{remark}
\begin{remark}
As in the case of the whole space $\R^3$, the distortion parameter $\lambda\neq 0$ in \eqref{eq:Veps}
is related to a point interaction with $\alpha\neq 0$. When $\lambda=0$, one obtains the resonant point interaction with $\alpha=0$ approximated by the pure scaling $
	V_{\eps}(x):=\frac{1}{\eps^2}V\left(\frac{x}{\eps}\right)$.
	\end{remark}

\begin{remark}
We restrict here to the three-dimensional case, but our arguments can be easily extended to bounded sets in dimension one and two. In the latter cases however the procedure of recovering $-\Delta_{\alpha,\sigma}$ by means of rescaled regular Schr\"odinger operators is easier than in $\R^3$, because spectral conditions and non generic potentials are not needed. For a different type of approximation in dimension two and three, we refer to \cite{BO23,Bo24}.
\end{remark}

\begin{remark}\label{re:re}
When the potential $V$ satisfies the assumptions of Theorem \ref{th:main}, we say that $-\Delta_{\sigma}+V$ has a (simple) zero-energy resonance. This notion is the localization to the bounded domain $\Omega$ of the analogous notion of resonance on $\R^3$ based on the Birman-Schwinger operator -- see Lemma \ref{le:local-res}. We point out that on $\R^3$ an equivalent notion of resonance (perhaps more common in the PDE framework) relies on the occurrence of non-$L^2$ solutions to $(-\Delta+V)\psi=0$. This notion is lost in the bounded case, and we refer to Section \ref{sec:res} for a more detailed discussion of this issue.
\end{remark}

The paper is organized as follows. In Section \ref{se2} we recall the construction and the main properties of the self-adjoint operators $-\Delta_{\sigma}$ and $-\Delta_{\alpha,\sigma}$. In Section \ref{se3}, after some preparatory technical facts, we prove our main result, Theorem \ref{th:main}. Section \ref{sec:res} is devoted to further comments and perspectives about the approximation mechanism of a point interaction in dimension three, with a focus on the role of spectral assumptions and the notion of zero-energy resonances.

\section{Point interaction on a bounded domain}\label{se2}
In this section we review the construction and the main properties of the self-adjoint operators $-\Delta_{\sigma}$ and $-\Delta_{\alpha,\sigma}$ on $L^2(\Omega)$. Consider the map
\begin{equation}\label{eq:def_sigma}
	\sigma(\psi):=(\partial_n\psi)_{|\partial\Omega}+b\,\psi_{|\partial\Omega},
\end{equation}
which in the limit case $b=+\infty$ is interpreted as
\begin{equation}\label{eq:def_sigmaD}
\sigma(\psi):=\psi_{|\partial\Omega}.
\end{equation}
The map $\sigma$ is well-defined and continuous from $H^2(\Omega)\to L^2(\Omega)$, and surjective onto $H^{1/2}(\partial\Omega)$ (respectively onto $H^{3/2}(\partial\Omega)$ in the Dirichlet case $b=+\infty$). Given $\psi\in H^2(\Omega)$, we can then write the Robin or Dirichlet boundary conditions \eqref{robin} and \eqref{D} as $\sigma(\psi)=0$. The Laplace operator $-\Delta_\sigma \psi=\Delta\psi $ with domain  
\begin{equation}\label{h2b}
		\mathcal D(-\Delta_\sigma):=H^2_{\sigma}(\Omega)=\{\psi\in H^2(\Omega)\,|\,\sigma(\psi)=0\},
\end{equation}
as it is well-known, is self-adjoint
(see for example \cite{ArendtPDE} for the general theory and \cite{BL07, BL12} for details in the framework of quasi-boundary triples). %Notice that this condition includes a boundary condition of mixed type, i.e. Dirichlet condition on $\partial\Omega_b$ and Neumann condition on $\partial\Omega_a$.  
%{\color{blue} Scrivere la definizione di $-\Delta_{\sigma}$, dire che \`e un operatore autoaggiunto, con dominio 
%	\begin{equation}\label{h2b}
%		H^2_{\sigma}(\Omega)=\{\psi\in H^2(\Omega)\,|\,\sigma(\psi)=0\}.
%	\end{equation}
%Aggiungere referenze.
%}

The Green function $\mathcal{G}^y_z(x)$ for $-\Delta_{\sigma}$ can be obtained as follows: for every $z\in\C\setminus\R$ and $y\in\Omega$, let $h_z^y$ be the solution to
\begin{equation}\label{def_dist}
\begin{cases}
(-\Delta+z)h_z^y=0,&x\in\Omega\\
\sigma(h_z^y)=-\sigma(\Gamma_z^y),&x\in\partial\Omega,
\end{cases}
\end{equation}
where $\Gamma^y_z(x):=\frac{e^{-\sqrt{z}|x-y|}}{4\pi|x-y|}$ is the Green function for the Laplacian on $\R^3$ (we use the branch of the square root such that $\RE\sqrt{z}>0$). Since $\Gamma_z^y$ is smooth on $\partial\Omega$, $h_z^y$ belongs to $H^2(\Omega)$, and $\mathcal{G}^y_z(x)$ is then given by
\begin{equation}\label{green_b}
\mathcal{G}^y_z(x)=\Gamma_z(x,y)+h_z^y(x).
\end{equation}
To lighten the notation, we shall write $\Gamma_z:=\Gamma^{0}_z$, $h_z^0=h_z$ and $\mathcal{G}_z:=\mathcal{G}^{0}_z$.

Consider now the symmetric, densely defined operator $S_{\sigma}$, given by
$$\mathcal{D}(S_{\sigma})=\{f\in H_{\sigma}^2(\Omega)\,|\,0\not\in\supp(f)\}, \quad S_{\sigma}f=-\Delta_{\sigma} f.$$
The operator $S_{\sigma}$ admits a one parameter family of self-adjoint extensions $-\Delta_{\alpha,\sigma}$ on $L^2(\Omega)$. The extension with $\alpha=\infty$ is $-\Delta_{\sigma}$, whilst the others represent non trivial point interactions at the origin, and they are given (see e.g.~\cite{BFM_bounded}, where the case of Dirichlet boundary condition is treated) by
\begin{gather}\label{def:da}
\mathcal{D}(-\Delta_{\alpha,\sigma})=\big\{\psi\in L^2(\Omega)\,|\,\psi=\phi_{z}+c_z(\alpha)\phi_z(0)\mathcal{G}_z,\,\phi_z\in H_{\sigma}^2(\Omega)\big\},\\
(-\Delta_{\alpha,\sigma}+z)\psi=(-\Delta_\sigma+z)\phi_z,
\end{gather}
for every fixed $z\in\C\setminus\R$, where
\begin{equation}\label{czalpha}
c_z(\alpha)=\left(\alpha+\frac{\sqrt{z}}{4\pi}-h_z(0)\right)^{-1}.
\end{equation}
The resolvent of $-\Delta_{\alpha,\sigma}$ is given, for $z\in\C\setminus\R$, by the explicit formula
\begin{equation}\label{res_form_bounded}(-\Delta_{\alpha,\sigma}+z)^{-1}=(-\Delta_{\sigma}+z)^{-1}+c_z(\alpha)\big|\mathcal{G}_z\big\rangle\big\langle\overline{\mathcal{G}_z}\big|.
\end{equation}
Now we show that, in a sense, the operator $-\Delta_{\alpha,\sigma}$ can be obtained by suitably restricting to $\Omega$ the analogous Schr\"odinger operator with point interaction $-\Delta_{\alpha}$ on $L^2(\R^3)$. First of all, let us recall (see e.g.~\cite{albeverio-solvable}) that we have the characterization
\begin{gather}\label{def:dr}
\mathcal{D}(-\Delta_{\alpha})=\big\{\psi\in L^2(\R^3)\,|\,\psi=\phi_{z}+d_z(\alpha)\phi_z(0)\Gamma_z,\,\phi_z\in H^2(\R^3)\big\},\\
(-\Delta_{\alpha}+z)\psi=(-\Delta+z)\phi_z,
\end{gather}
for every fixed $z\in\C\setminus\R$, where
$$d_z(\alpha):=\left(\alpha+\frac{\sqrt{z}}{4\pi}\right)^{-1}.$$
Again, the case $\alpha=+\infty$ corresponds to the unperturbed Laplace operator $-\Delta$ on $L^2(\R^3)$. The resolvent of $-\Delta_{\alpha}$ is given, for $z\in\C\setminus\R$, by the explicit formula
\begin{equation}\label{eq:ress}
(-\Delta_{\alpha}+z)^{-1}=(-\Delta+z)^{-1}+d_z(\alpha)\big|\Gamma_z\big\rangle\big\langle\overline{\Gamma_z}\big|.
\end{equation}

\section{Proof of the main result}\label{se3}
In this section we prove our main result, Theorem \ref{th:main}, on the approximation of a point interaction on bounded domains. We exploit a suitable extension-restriction argument, which allows to reduce to the corresponding result on $\R^3$. 

 We start with the following technical lemma, which shows that the self-adjoint operator $-\Delta_{\alpha}$ on $L^2(\R^3)$ can be obtained as a suitable extension of the corresponding operator $-\Delta_{\alpha,\sigma}$ on $L^2(\Omega)$.
\begin{lemma}\label{le:ext}
	There exists an extension operator $\mathcal{D}(-\Delta_{\alpha,\sigma})\ni \psi\mapsto\widetilde{\psi}\in\mathcal{D}(-\Delta_{\alpha})$ such that $\widetilde{\psi}_{|\Omega}=\psi$ and
	$$(\Delta_{\alpha}\widetilde{\psi}\,)_{|\Omega}=\Delta_{\alpha,\sigma}\psi,\qquad\|\Delta_{\alpha}\widetilde{\psi}\|_{L^2(\R^3)}\lesssim \|\Delta_{\alpha,\sigma}\psi\|_{L^2(\Omega)},$$
	for every $\psi\in\mathcal{D}(-\Delta_{\alpha,\sigma})$.
\end{lemma}
\begin{proof}
	Let us fix $z\in \C\setminus\R$, and according to \eqref{def:da} we write
	$$\psi=\phi_{z}+c_z(\alpha)\phi_z(0)\mathcal{G}_z.$$
	Using the representation \eqref{green_b} for the Green function $\mathcal{G}_z$, we can rewrite
	$$\psi=\phi_{z}+c_z(\alpha)\phi_z(0)h_z+c_z(\alpha)\phi_z(0)\Gamma_z.$$
	Consider an $H^2(\R^3)$ extension $\widetilde{\phi}_z$ of $\phi_z$ (see e.g.~\cite{ArendtPDE}), with $\|\widetilde{\phi}_z\|_{H^2(\R^3)}\lesssim \|\phi_z\|_{H^2(\Omega)}$, and set $\widetilde{F}_z:=\widetilde{\phi}_z+c_z(\alpha)\phi_z(0)\widetilde{h}_z$, where $\widetilde{h}_z$ is an $H^2(\R^3)$ extension of $h_z$. In view of the continuous embedding $H^2(\R^3)\hookrightarrow \mathcal{C}_b(\R^3)$, we have $\|\widetilde{F}_z\|_{H^2(\R^3)}\lesssim\|\phi_z\|_{H^2(\Omega)}$. Next, let us set $\widetilde{\psi}:=\widetilde{F}_z+c_z(\alpha)\phi_z(0)\Gamma_z$, and observe that by construction $\widetilde{\psi}_{|\Omega}=\psi$. Moreover, a direct computation shows that
	$$c_z(\alpha)\phi_z(0)=d_z(\alpha)F_z(0).$$
	Hence, we deduce from representation \eqref{def:dr} that $\widetilde{\psi}\in\mathcal{D}(-\Delta_{\alpha})$. In particular,
	$$\big((-\Delta_{\alpha}+z)\widetilde{\psi}\big)_{|\Omega}=\big((-\Delta+z)\widetilde{F}_z\big)_{|\Omega}=(-\Delta_{\sigma}+z)\phi_z=(-\Delta_{\alpha,\sigma}+z)\psi,$$
	and similarly
	\begin{equation*}
		\begin{split}
			\|(-\Delta_{\alpha}+z)\widetilde{\psi}\|_{L^2(\R^3)}&=\|(-\Delta+z)\widetilde{F}_z\|_{L^2(\R^3)}\lesssim \|\widetilde{F}_z\|_{H^2(\R^3)}\lesssim\|\phi_z\|_{H^2(\Omega)}\\
			&\lesssim\|(-\Delta_{\sigma}+z)\phi_z\|_{L^2(\Omega)}=\|(-\Delta_{\alpha,\sigma}+z)\psi\|_{L^2(\Omega)},
		\end{split}
	\end{equation*}
	which proves the claim.
\end{proof}

As a second preliminary fact, we recall the standard approximation theory for a point interaction on $\R^3$, with a further information on the convergence topology. To this end, we introduce the map $$\widetilde{B}_0f:=u(-\Delta)^{-1}vf,$$ where it is understood that $(-\Delta)^{-1}=\frac{1}{|x|}*$ is the convolution with the Newtonian potential and $u,\ v$ are defined as in \eqref{uv}. In view of the assumptions on $V$, $\widetilde{B}_0$ defines a compact operator on $L^2(\R^3)$.

\begin{theorem}\label{conv_tre}
Suppose that $-1$ is a simple eigenvalue for $\widetilde{B}_0$, and let $\widetilde{\varphi}$ be the associated normalized eigenfunction. Suppose moreover that $\langle v,\widetilde{\varphi}\rangle\neq 0$, and set $\alpha=-\lambda|\langle v,\widetilde{\varphi}\rangle|^{-2}\neq +\infty$. The following facts hold.
\begin{itemize}
	\item[(i)] $-\Delta+V_{\eps}$ converges to $-\Delta_{\alpha}$ as $\eps\to 0$, in the norm resolvent sense in $L^2(\R^3)$.	
	\item[(ii)] Fix an open, bounded set $\Theta\subseteq\R^3$, with $0\not\in\overline{\Theta}$. For every $z\in\C\setminus\R$, $(-\Delta+V_{\eps}+z)^{-1}$ converges as $\eps\to 0$ to $(-\Delta_{\alpha}+z)^{-1}$, in the norm operator topology of $\mathcal{B}(L^2(\R^3);H^2(\Theta))$.
\end{itemize}
\end{theorem}

\begin{proof}[Proof of Theorem \ref{conv_tre}]
Part (i) is a particular instance of the general approximation result in $\R^3$ \cite[Theorem 1.2.5]{albeverio-solvable}, to which we refer. %which also allows for non-compactly supported potentials $V$ (with suitable decay assumptions at infinity), and includes the case when $-1$ is either not an eigenvalue or a degenerate eigenvalue for $\widetilde{B}_0$. See also the recent critical analysis in \cite{DMV24}. 
In order to prove part (ii), we recall that (see formula (1.2.16) in \cite{albeverio-solvable})
\begin{equation}\label{kku}
(-\Delta+V_{\eps}+z)^{-1}=(-\Delta+z)^{-1}-(1+\lambda)A_{\eps}(z)\eps[1+B_{\eps}(z)]^{-1}C_{\eps}(z),
\end{equation}
where
\begin{equation*}
	\begin{split}
(A_{\eps}(z)f)(x)&:=\int_{\R^3}\Gamma_z^{\eps y}(x)v(y)f(y)dy\\
(B_{\eps}(z)f)(x)&:=(1+\lambda\eps)u(x)\int_{\R^3}\Gamma_{\eps z}^{y}(x)v(y)f(y)dy\\
(C_{\eps}(z)f)(x)&:=u(x)\int_{\R^3}\Gamma^y_z(\eps x)f(y)dy,
\end{split}
\end{equation*}
The limits of $A_{\eps}(z)$, $\eps[1+B_{\eps}(z)]^{-1}$ and $C_{\eps}(z)$ as $\eps\to 0$ exist (in particular when $\langle v,\widetilde{\varphi}\rangle=0$ the operator $[1+B_{\eps}(z)]$ has a bounded inverse, and the r.h.s.~of \eqref{kku} produces the free resolvent in the limit). %as $\eps\to 0$). 
Now, the following fact holds true:
\begin{itemize}
\item[(*)] $A_{\eps}(z)$, that is the operator with integral kernel $\Gamma_z^{\eps y}(x)v(y)$, converges as $\eps\to 0$, in the norm operator topology of $\mathcal{B}(L^2(\R^3);H^2(\Theta))$, to the operator $A$ with integral kernel $\Gamma_z(x)v(y)$.
\end{itemize}
To this end, observe that for $\eps$ small enough we have
$$|\Gamma_z^{\eps y}(x)-\Gamma_z(x)|\lesssim\eps,$$
uniformly for $x\in\Theta$ and $y\in\supp(V)$. Hence
\begin{align*}
&\|A_{\eps}f-Af\|_{H^2(\Theta)}=\Big\|\int_{\R^3}\big(\Gamma_z^{\eps y}-\Gamma_z\big)v(y)f(y)dy\Big\|_{H^2_x(\Theta)}\\
&\quad\lesssim \int_{\supp(V)}\|\Gamma_z^{\eps y}-\Gamma_z\|_{L^2_x(\Theta)}v(y)|f(y)|dy\lesssim\eps\|V\|^{1/2}_{L^1(\R^3)}\|f\|_{L^2(\R^3)},
\end{align*}
which proves (*), as the hypothesis that $V$ is compactly supported guarantees in particular $\|V\|_{L^1(\R^3)}\lesssim\|V\|_{L^{3/2}(\R^3)}$.
\end{proof}

%---------------

Next we show that eigenfunctions for the operator $\widetilde{B}_0$ on $L^2(\R^3)$ localize to eigenfunctions for $B_0$ on $L^2(\Omega)$.

\begin{lemma}\label{le:local-res}
The following facts hold. 
\begin{itemize}
	\item[(i)] Suppose that $\varphi\in L^2(\Omega)$ solves $B_0\varphi=-\varphi$, and let $\widetilde{\varphi}$ be the extension by zero of $\varphi$ outside $\Omega$. Then $\widetilde{\varphi}$ satisfies $\widetilde{B}_0\widetilde{\varphi}=-\widetilde{\varphi}$.
	\item[(ii)] Conversely, if $\widetilde{\varphi}\in L^2(\R^3)$ solves $\widetilde{B}_0\widetilde{\varphi}=-\widetilde{\varphi}$, then $\widetilde{\varphi}=0$ on $\Omega^c$, and $\varphi:=\widetilde{\varphi}_{|\Omega}$ satisfies $B_0\varphi=-\varphi$.
\end{itemize}
	\end{lemma}

\begin{proof}
Both assertions readily follow from the definition \eqref{defB} of $B_0$, the identity
$$\widetilde{B}_0\widetilde{\varphi}(x)=\frac{u(x)}{4\pi}\int_{\R^3}\frac{v(y)\widetilde{\varphi}(y)}{|x-y|}dy,$$ and the fact that $\operatorname{supp}(u)=\operatorname{supp}(v)\subseteq\Omega$.
\end{proof}

Our last preparatory result provides a resolvent estimate for the approximating operator $-\Delta_{\sigma}+V$, in a topology that will be suitable to manage traces on the boundary $\partial\Omega$.

\begin{lemma}\label{int_le}
	Let $V\in L^{\frac32}(\R^3)$, with $\operatorname{supp}(V)\subseteq\Omega$. Let us fix moreover an open ball $B$ such that $\supp(V)\subsetneq\overline{B}\subseteq\Omega$, and set $\Theta:=\Omega\setminus\overline{B}$. Then
	\begin{equation}\label{est_nor}
		\|(-\Delta_{\sigma}+V+z)^{-1}\|_{\mathcal{B}(L^2(\Omega);H^2(\Theta))}\lesssim \|V\|_{L^{3/2}}.
	\end{equation}
\end{lemma}

\begin{proof}
We start by showing that $u(-\Delta_{\sigma}+z)^{-1}v$ is an Hilbert-Schdmit operator (whence bounded) on $L^2(\Omega)$, with norm controlled by $\|V\|_{L^{3/2}}$. Indeed
$$\|u(-\Delta_{\sigma}+z)^{-1}v\|^2_{\mathrm{H.S.}}\lesssim\int_{\Omega\times\Omega}|V(x)|\big(|\Gamma_z(x,y)|^2+|h_{z}^y(x)|^2\big)|V(y)|dxdy,$$
and using Hardy-Littlewood-Sobolev and H\"older inequality we get
\begin{equation}\label{ua}
	\int_{\Omega\times\Omega}|V(x)||\Gamma_z(x,y)|^2|V(y)|dxdy\lesssim\|V\|^2_{L^{\frac32}}.
\end{equation}
Observe moreover that, since $\Gamma_z$ is smooth outside $x=0$, for every $y\in\supp(V)$ we have the bound $\|\sigma(\Gamma^y_z)\|_{L^2(\partial\Theta)}\leqslant\delta$, for some positive constant $\delta$ independent on $y$. In view of the definition \eqref{def_dist} of $h_z^y$, we then have
\begin{equation}\label{est_hh}
	\|h_z^y\|_{H^2_x(\Omega)}\lesssim 1.
\end{equation}
In particular, $|h_z^y(x)|$ is uniformly bounded for $x,y\in \Omega$, 
%In particular, $\|h_z^y(x)\|_{L^{\infty}_{x,y}}\lesssim 1$ {\color{red}(Definire $L^{\infty}_{x,y}$)}, 
which implies
\begin{equation}\label{pu}
	\int_{\Omega\times\Omega}|V(x)||h_z^y(x)|^2|V(y)|dxdy\lesssim\|V\|^2_{L^1(\Omega)}\lesssim\|V\|^2_{L^{3/2}(\Omega)}.
\end{equation}
Combining estimates \eqref{ua} and \eqref{pu} we eventually obtain
\begin{equation}\label{due}
\|u(-\Delta_{\sigma}+z)^{-1}v\|^2_{\mathrm{H.S.}}\lesssim \|V\|_{L^{3/2}}.
\end{equation}
Estimate \eqref{due} and the condition $|V|^{\frac12}\lesssim(-\Delta)^{\frac12}$ (in the sense of infinitesimally bounded operators), which can be proved as in \cite[Lemma 4.1(ii)]{Mic-Scan}, guarantee the validity of the Konno-Kuroda identity \cite[Theorem B.1(b)]{albeverio-solvable}
\begin{align*}
	&(-\Delta_{\sigma}+V+z)^{-1}=(-\Delta_{\sigma}+z)^{-1}\\
	&-(-\Delta_{\sigma}+z)^{-1}v\big(1+u(-\Delta_{\sigma}+z)^{-1}v\big)^{-1}u(-\Delta_{\sigma}+z)^{-1}.
\end{align*}
Moreover we have the following bounds:
	\begin{gather}
		\label{uno}\|u(-\Delta_{\sigma}+z)^{-1}\|_{\mathcal{B}(L^2(\Omega);L^2(\Omega))}\lesssim\|u\|_{L^2(\Omega)},\\
	\label{tre}\|(-\Delta_{\sigma}+z)^{-1}v\|_{\mathcal{B}(L^2(\Omega);H^2(\Theta))}\lesssim\|v\|_{L^2(\Omega)}.
	\end{gather}
	The bound \eqref{uno} follows from the Sobolev embedding $H^2(\Omega)\hookrightarrow L^{\infty}(\Omega)$. Moreover, using \eqref{est_hh} and the estimate $\|\Gamma^y_z\|_{H^2(\Theta)}\lesssim 1$ for $y\in\supp(V)$, we obtain
	\begin{align*}
		\|(-\Delta_{\sigma}&+z)^{-1}vf\|_{H^2(\Theta)}=\Big\|\int_{\Omega}\mathcal{G}^y_z\,v(y)f(y)dy\Big\|_{H^2(\Theta)}\\
		&\lesssim\int_{\Omega}\big\|\Gamma^y_z+h_{z}^y\big\|_{H^2(\Theta)}v(y)f(y)dy\lesssim \|v\|_{L^2(\Omega)}\|f\|_{L^2(\Omega)},
	\end{align*}
	which proves the bound \eqref{tre}. Combining \eqref{uno}, \eqref{due} and \eqref{tre}, we eventually deduce estimate \eqref{est_nor}.
\end{proof}

We are now ready to prove our main result. Given a function $\widetilde{\psi}\in H^2(\R^3)$, with a slight abuse of notation we will write $\sigma(\widetilde{\psi}):=\sigma(\widetilde{\psi}_{|\Omega})$, where $\sigma:H^2(\Omega)\to L^2(\partial\Omega)$ is the continuous map defined by \eqref{eq:def_sigma} (respectively by \eqref{eq:def_sigmaD} in the case of a Dirichlet boundary condition on $\partial\Omega$).

\begin{proof}[Proof of Theorem \ref{th:main}]
Fix $z\in\C\setminus\R$ and $f\in L^2(\Omega)$, with $\|f\|_{L^2(\Omega)}=1$. We set $$\psi:=(-\Delta_{\alpha,\sigma}+z)^{-1}f,\qquad \widetilde{f}:=(-\Delta_{\alpha}+z)\widetilde{\psi},$$ 
where $\widetilde{\psi}\in\mathcal{D}(-\Delta_{\alpha})$ is the extension of $\psi$ provided by Lemma \ref{le:ext}, so that $$\widetilde{f}_{|\Omega}=f,\qquad  \|\widetilde{f}\|_{L^2(\R^3)}\lesssim 1.$$
From now on, to ease the notation, we write $o(1)$ to denote a quantity \emph{independent} on $f$ and infinitesimal as $\eps\to 0$. Our aim is to show that 
\begin{equation}\label{eq:conv_res_omega}
\|(-\Delta_{\sigma}+V_{\eps}+z)^{-1}f-\psi\|_{L^2(\Omega)}=o(1).
\end{equation}
We proceed as follows. Let $\widetilde{f}_{\eps}\in L^2(\R^3)$ be such that
\begin{equation}\label{corr:V}
	(-\Delta_{\sigma}+V_{\eps}+z)^{-1}f=\big((-\Delta+V_{\eps}+z)^{-1}\widetilde{f}_{\eps}\,\big)_{|\Omega}.
\end{equation}
We point out that, in general, $\widetilde{f}_{\eps}\neq\widetilde{f}$, in view of the correction due to the boundary condition on $\partial\Omega$. Nevertheless, we are going to show that $\|\widetilde{f}_{\eps}-\widetilde{f}\|_{L^2(\R^ 3)}=o(1)$. With this information, we can conclude the proof as follows. Owing to Lemma \ref{le:local-res}, the operator $-\Delta+V$ on $L^2(\R^3) $ satisfies the assumptions of Theorem \ref{conv_tre} (with $\widetilde{\varphi}$ be the extension by zero of $\varphi$ outside $\Omega$), whose norm-convergence result in point (i) then gives
\begin{equation}\label{eq:conv_res_full} 
\|(-\Delta+V_{\eps}+z)^{-1}\widetilde{f}_{\eps}-\widetilde{\psi}\|_{L^2(\R^3)}=o(1).
\end{equation}
Restricting \eqref{eq:conv_res_full} to $L^2(\Omega)$, and using \eqref{corr:V}, we eventually deduce \eqref{eq:conv_res_omega}.

We construct now $\widetilde{f}_{\eps}$, showing indeed that $\|\widetilde{f}_{\eps}-\widetilde{f}\|_{L^2(\R^3)}=o(1)$. Let us fix an open ball $B$ such that $\supp(V)\subsetneq\overline{B}\subseteq\Omega$, and set $\Theta:=\Omega\setminus\overline{B}$. By point (ii) of Theorem \ref{conv_tre}, we have
\begin{equation}\label{eq:conv_bound_theta}\|(-\Delta + V_{\eps}+z)^{-1}\widetilde{f}-\widetilde{\psi}\|_{H^2(\Theta)}=o(1).
\end{equation}
Let $\mathcal{X}_{\sigma}=\sigma(H^2(\Omega))$, that is $H^{3/2}(\partial\Omega)$ when $b=\infty$, $H^{1/2}(\partial\Omega)$ otherwise. Set also
$$\phi_{\eps}:=\sigma\big((-\Delta + V_{\eps}+z)^{-1}\widetilde{f}\,\big).$$
Since by construction $\psi\in\mathcal{D}(-\Delta_{\alpha,\sigma})$, we have $\sigma(\widetilde{\psi})=0$. We then deduce from \eqref{eq:conv_bound_theta} that $\|\phi_{\eps}\|_{\mathcal{X}_{\sigma}}=o(1)$. Next, let us set $\widetilde{f}_{\eps}:=\widetilde{f}+\widetilde{r}_{\eps}$, where $\widetilde{r}_{\eps}$ solves
\begin{equation}\label{eq:sire}
\begin{cases}
\sigma((-\Delta+V_{\eps}+z)^{-1}\widetilde{r}_{\eps})=-\phi_{\eps},\\
\widetilde{r}_{\eps}(x)=0\quad x\in\Omega.
\end{cases}
\end{equation}
Observe that $\widetilde{f}_{\eps}$ satisfies identity \eqref{corr:V}. We are left to show that we can construct $\widetilde{r}_{\eps}$ so that $\|r_{\eps}\|_{L^2}=o(1)$. Since $\sigma$ is surjective onto $\mathcal{X}_{\sigma}$, by the open mapping theorem there exists $h_{\eps}\in H^2(\Omega)$ with $\sigma(h_{\eps})=-\phi_{\eps}$, $h_{\eps|\supp(V)}=0$, and $$\|h_{\eps}\|_{H^2(\Omega)}\lesssim \|\phi_{\eps}\|_{L^2(\partial\Omega)}=o(1).$$
Let us define the functions
$$\ell_{\eps}:=(-\Delta+V_{\eps}+z)h_{\eps}=(-\Delta+z)h_{\eps},\qquad k_{\eps}:=(-\Delta_{\sigma}+V_{\eps}+z)^{-1}\ell_{\eps}.$$
Using Lemma \ref{int_le} and the uniform bound $\|V_{\eps}\|_{L^{\frac32}}\lesssim\|V\|_{L^{\frac32}}$ we obtain
\begin{equation}\label{est_norg}
\|(-\Delta_{\sigma}+V_{\eps}+z)^{-1}\|_{\mathcal{B}(L^2(\Omega);H^2(\Theta))}\lesssim 1.
\end{equation}
Then we get
\begin{equation}\label{est_k}
\|k_{\eps}\|_{H^2(\Theta)}\lesssim\|\ell_{\eps}\|_{L^2(\Omega)}\lesssim\|h_{\eps}\|_{H^2(\Omega)}=o(1).
\end{equation}
Now, let us set $w_{\eps}:=h_{\eps}-k_{\eps}$, and observe that \eqref{est_k} yields $\|w_{\eps}\|_{H^2(\Theta)}=o(1)$. Moreover, $w_{\eps}$ satisfies
$$\begin{cases}
(-\Delta+V_{\eps}+z)w_{\eps}=0&x\in\Omega\\
\sigma(w_{\eps})=-\phi_{\eps}.
\end{cases}$$
Let $\widetilde{w}_{\eps}\in H^2(\R^3)$ be an extension of $w_{\eps}$, with $\|\widetilde{w}_{\eps}\|_{H^2(\Omega^c)}\lesssim\|w_{\eps}\|_{H^2(\Theta)}$, and set $\widetilde{r}_{\eps}:=(-\Delta+V_{\eps}+z)\widetilde{w}_{\eps}$. By construction $\widetilde{r}_{\eps}$ satisfies \eqref{eq:sire}, and 
\begin{equation}\label{r.conv}
\|\widetilde{r}_{\eps}\|_{L^2(\R^3)}=\|\widetilde{r}_{\eps}\|_{L^2(\Omega^c)}\lesssim\|\widetilde{w}_{\eps}\|_{H^2(\Omega^c)}\lesssim\|w_{\eps}\|_{H^2(\Theta)}=o(1),
\end{equation}
as desired. The proof is complete.
\end{proof}

\section{Further comments and developments}\label{sec:res}
%In this final section we provide a brief overview on the approximation mechanisms of a point interaction in dimension three, with a focus on the meaning of the spectral condition required. 

%We start by recalling the well-established situation in the whole $\R^3$ setting. As before, let us fix a compactly supported $V\in L^{\frac32}(\R^3)$ and set $u:=\operatorname{sgn}(V)|V|^{1/2}$, $v:=|V|^{1/2}$, and $\widetilde{B}_0:=u(-\Delta)^{-1}v$. 
\subsection{Different concepts of resonance} With reference to the notation of Section \ref{se3}, we recall the following definition.

\begin{definition}\label{def_res_1}
We say that $-\Delta+V$ has a zero-energy resonance if there exists $\widetilde{\varphi}\in L^2(\R^3)$, with $\langle v,\widetilde{\varphi}\rangle\neq 0$, such that $\widetilde{B}_0\widetilde{\varphi}=-\widetilde{\varphi}$.
\end{definition}

Theorem 1.2.5 in \cite{albeverio-solvable} shows that $-\Delta+V_{\eps}$ converges as $\eps\to 0$, in norm resolvent sense, to $-\Delta_{\alpha}$ for some $\alpha\neq\infty$ (i.e.~to a non trivial point interaction) if and only if $-\Delta+V$ has a zero-energy resonance according to the previous definition. 
%namely if $-1$ is an eigenvalue of $\widetilde{B}_{0}$, with at least one eigenfunction $\widetilde{\varphi}$ satisfying the non-orthogonality condition $\langle v,\widetilde{\varphi}\rangle\neq 0$.

\begin{remark}
Potentials satisfying the zero-energy resonance condition do actually exist, see e.g.~\cite[Proposition 7.1]{Mic-Scan}, even though this occurrence is, in a sense, ``non-generic". Given indeed $V\in L^{3/2}(\R^3)$, consider the family of potentials $\{\theta V\}_{\theta\in\R}$: since $\widetilde{B}_0$ is a compact operator on $L^2(\R^3)$, then $-1$ is an eigenvalue of $\theta\widetilde{B}_0$ for at most a countable set of coupling parameters $\theta$.
\end{remark}

A different notion of zero-energy resonance, widely used in the PDE's context, is the following.

\begin{definition}\label{def_res_2}
	We say that $-\Delta+V$ has a zero-energy resonance if there exists $\widetilde{\psi}\in L^2_{-1-\delta}(\R^3)\setminus L^2(\R^3)$, for any $\delta>0$, such that $(-\Delta+V)\widetilde{\psi}=0$ as a distributional identity on $\R^3$, where $L^2_{s}(\R^3)$, $s\in\R$, denote the weighted $L^2$-space $L^2(\R^3,\langle x\rangle^{s}dx).$
\end{definition}

The equivalence between Definitions \ref{def_res_1} and \ref{def_res_2} is encoded by the following result, which is an instance of the celebrated Birman-Schwinger principle. \cite{Birman,Sch}. 
\begin{lemma}\label{le:birs}
	Let $V\in L^{\frac32}(\R^3)$ be a compactly supported potential.
	\begin{itemize}
		\item[(i)] Let $\widetilde{\varphi}\in L^2(\R^3)$ be such that $B_0\widetilde{\varphi}=-\widetilde{\varphi}$. Then the function \begin{equation}\label{rel}\widetilde{\psi}=(-\Delta)^{-1}v\widetilde{\varphi}\in L_{-1-\delta}^2(\R^3),\;\delta>0,\end{equation}
		satisfies $(-\Delta+V)\widetilde{\psi}=0$.
		\item[(ii)] Conversely, if $\widetilde{\psi}\in L_{-1-\delta}^2(\R^3)$ for any $\delta>0$, solves $(-\Delta+V)\widetilde{\psi}=0$, then the function
		$$\widetilde{\varphi}=u\widetilde{\psi}\in L^2(\R^3)$$
		satisfies $\widetilde{B}_0\widetilde{\varphi}=-\widetilde{\varphi}$.
	\end{itemize}
	%\begin{itemize}
	%	\item[(i)] there exists a non-zero $\varphi\in L^2(\Omega)$ such that $B_0\varphi=-\varphi$,
	%	\item[(ii)] there exists a non-zero $\psi\in L^2(\R^3,\langle x\rangle^{-1-\delta})$ such that $(-\Delta+V)\psi=0$.
	%\end{itemize}
	%If the conditions above are satisfied, $\psi$ and $\varphi$ are related by 
	%$$\psi=(-\Delta)^{-1}v\widetilde{\varphi},$$
	%where $\widetilde{\varphi}\in L^2(\R^3)$ is the extension to zero of $\varphi$ outside $\Omega$. 
	In both cases, $\langle v,\widetilde{\varphi}\rangle\neq 0$ if and only if $\widetilde{\psi}\not\in L^2(\R^3)$.
\end{lemma}

%\begin{proof}
%	\textbf{(i)} Let $\varphi\in L^2(\Omega)\setminus\{0\}$, with $B_0\varphi=-\varphi$. We have
%	\begin{equation}\label{eq:int_int}
	%		u(x)\int_{\R^3}\frac{v(y)\widetilde{\varphi}(y)}{|x-y|}dy=-\widetilde{\varphi}(x)
	%	\end{equation}
%	as an identity in $L^2(\R^3)$. Indeed, \eqref{eq:int_int} reduces to $B_0\varphi=-\varphi$ for $x\in\Omega$, and it is trivially satisfied on $\Omega^{c}$, since $\operatorname{supp}(u)\subseteq\Omega$. Identity \eqref{eq:int_int} can be equivalently written as $u(-\Delta)^{-1}v\widetilde{\varphi}=-\widetilde{\varphi}$. Then we get
%	$$-\Delta\widetilde{\psi}=v\widetilde{\varphi}=-uv(-\Delta)^{-1}v\widetilde{\varphi}=V\widetilde{\psi},$$
%	that is, $(-\Delta+V)\widetilde{\psi}=0$.\smallskip

%	\noindent\textbf{(ii)} Let $\widetilde{\varphi}=u\widetilde{\psi}\in L^2(\R^3)$. We have
%	$$u(-\Delta)^{-1}v\widetilde{\varphi}=u(-\Delta^{-1})(V\widetilde{\psi})=-u(-\Delta)^{-1}(-\Delta)\widetilde{\psi}=-\widetilde{\varphi}.$$
%	Restricting the identity above to $\Omega$, we get $B_0(\varphi)=-\varphi$, as desired. \smallskip

%	Now, let us write $$\widetilde{\psi}=(-\Delta)^{-1}v\widetilde{\varphi}=\frac{\langle v,\varphi\rangle}{4\pi|x|}+\widetilde{\psi}_1,$$
%	where 
%	$$\widetilde{\psi}_1=\frac{1}{4\pi}\int_{\R^3}\Big(\frac{1}{|x-y|}-\frac{1}{|x|}\Big)v(y)\widetilde{\varphi}(y)dy.$$
%	We have $\widetilde{\psi}_1\in L^2(\R^3)$ {\color{blue}(aggiungi conto)}. Since $\frac{1}{|x|}\not\in L^2(\R^3)$, we then have $\widetilde{\psi}\in L^2(\R^3)$ if and only if $\langle v,\varphi\rangle\neq 0$, which concludes the proof.
%\end{proof}
Versions of the above result can be found e.g.~in \cite[Lemma 1.2.3]{albeverio-solvable} and in \cite[Proposition 4.3]{DEPA_arx}. The latter reference contains also a comprehensive discussion, including other equivalent notions of zero-energy resonance.

\begin{remark} According to definition \ref{def_res_2}, $\widetilde{\psi}$ formally satisfies the eigenvalue equation with zero eigenvalue, but $\psi\notin L^2$ and it is not related to the continuous spectrum. More precisely, the so called \emph{resonance function} $\widetilde{\psi}$ decays at infinity exactly as the Green function $\Gamma_0\in L^2_{-1-\delta}(\R^3)$ of the free Laplacian. We point out (see e.g.~the seminal paper \cite{J-K} by Jensen and Kato) that the resolvent $(-\Delta+V-z^2)^{-1}$, $z\in\C^+$, considered as a map from $L^2_{2+\sigma}$ to $L^2_{-2-\sigma}$, $\sigma>0$, has a continuous extension to the real line, except for a possible singularity at $z=0$, which corresponds to the occurrence of a zero-energy obstruction (eigenvalue or resonance). An analogous result holds true also for Schr\"odinger operators with point interactions \cite{scan_exp,Mic-Scan-reso}. Finally, we mention that a similar idea can be also extended to the context of non-selfadjoint operators, leading to the notion of virtual levels (see e.g.~\cite{BC} and references therein).
\end{remark}
We stress that this second notion of resonance, as such,  is clearly void in the case of bounded sets. As shown in Lemma \ref{le:local-res}, the notion of zero-energy resonance provided by Definition \ref{def_res_1} naturally ``localizes" from $\R^3$ to $\Omega$, as the non-local operator $(-\Delta)^{-1}$ is sandwiched between the compactly supported functions $u,v$. Starting from this, and using an extension-restriction approach, we have shown in Theorem \ref{th:main} the approximation procedure of $-\Delta_{\sigma,\alpha}$  by means of local Schr\"odinger operators on $L^2(\Omega)$. Still, it is unclear whether there is a localized version of Definition \ref{def_res_2}. Note indeed that, although $V$ is compactly supported, $\widetilde{\psi}$ is not in general, and we do not have a canonical way to describe its behavior on $\partial\Omega$, and to relate it to the eigenfunction $\varphi$ of $B_0$ as identity \eqref{rel} does in the whole space case. The possibility of replacing Definition \ref{def_res_2} of zero-energy resonance for the Schr\"odinger operator $-\Delta_{\sigma}+V$ with a substitute adapted to the presence of boundaries remains a question to investigate. 

\subsection{Non-local approximation.}
So far, we have considered the approximation of a point interaction on $\R^3$ by means of \emph{local} operators, focusing on the spectral conditions on the potential $V$ needed to ensure the emergence of a non-trivial boundary condition in the limit.

An alternative approximation scheme is based on \emph{non-local} operators. More precisely, we consider the following rank-one perturbation of the free Laplacian
$$H^\eps=-\Delta + a(\eps)(\rho^\eps,\cdot) \rho^\eps,$$
for some function (say, smooth and compactly supported) $\rho$ with $\int_{\R^3} \rho(x)\ dx =1$, so that  $\rho^\eps(x)=\frac{1}{\eps^3}\, \rho \! \left(\frac{x}{\eps} \right)$ is a sequence approximating the $\delta$ distribution on $\R^3$, and $a(\eps)$ is at this stage unprejudiced. $H^\eps$ is a family of self-adjoint operators on $L^2(\R^3)$ and a direct analysis shows that $H^\eps$ converges (in the norm resolvent sense as $\eps\to 0$) to $-\Delta_\alpha$ if and only if 
\begin{equation}\label{aeps}
	a(\eps)=-\dfrac{\eps}{\ell} + \dfrac{\alpha\eps^2}{\ell^2}+o(\eps^2),
	\end{equation}
where $\ell=(\rho,(-\Delta)^{-1}\rho) $ is the electrostatic energy of $\rho$, while any other scaling makes the family convergent to the free Laplacian (see \cite{NP1, NP2} for a detailed analysis with applications to the classical Pauli-Fierz model in electrodynamics, \cite{MP24} for a related scalar model in acoustics, and \cite{CFNT17} for a nonlinear Schr\"odinger model).
Now, also in this non-local approximation of $\Delta_\alpha$, a resonance function appears, given by the function $\psi=\Gamma_0*\rho\in L^2_{-1-\delta}(\R^3)\setminus L^2(\R^3)$, solving the equation $-\Delta\psi   -
\frac{1}{\ell}  \rho  \langle\rho,\psi\rangle=0$. Notice however that, within the non-local approximation scheme, the zero energy resonance always exists, and so there is no need of any spectral condition.  On the other hand, it is unclear if resonances in the sense of Definition \ref{def_res_1} make sense in this approximation scheme.\\
Let us turn now our attention to the case of a bounded domain $\Omega\subseteq\R^3$. 
Concerning the approximation through non-local operators, such mechanism also persists to the bounded case. Indeed, considering the self-adjoint operator 
$$H_{\sigma}^{\eps}=-\Delta_{\sigma}+a(\eps)(\rho^\eps,\cdot) \rho^\eps,$$
on $L^2(\Omega)$, with $\sigma$ the homogeneous Robin or Laplace boundary condition on $\partial\Omega$ and $\rho$ as before (and such that $\operatorname{supp}(\rho)\subseteq \Omega$), a direct inspection shows that $H_{\sigma}^{\eps}$ converges (in the norm resolvent sense as $\eps\to 0$) to $-\Delta_{\alpha,\sigma}$ with $\alpha\in\R$ (i.e.~a non trivial point interaction), if and only if $a(\eps)$ has the form \eqref{aeps}. For the sake of completeness, let us outline the argument here.

We start with the resolvent identity
\begin{equation*}(H^{\eps}_{\sigma}+z)^{-1}=(-\Delta_{\sigma}+z)^{-1}-\frac{\langle(-\Delta_{\sigma}+\overline{z})^{-1}\rho^{\eps},\cdot\rangle(-\Delta_{\sigma}+z)^{-1}\rho^{\eps}}{\frac{1}{a(\eps)}+\langle \rho^{\eps},(-\Delta_{\sigma}+z)^{-1}\rho^{\eps}\rangle},
\end{equation*}
valid for $z\in\C\setminus\R$. Since $\rho^\eps\to\delta$, we deduce
$$\langle(-\Delta_{\sigma}+\overline{z})^{-1}\rho^{\eps},\cdot\rangle(-\Delta_{\sigma}+z)^{-1}\rho^{\eps}\to |\mathcal{G}_z\rangle\langle \overline{\mathcal{G}_z}|$$
in the norm resolvent sense as $\eps\to 0$.
Recalling identities \eqref{res_form_bounded} and  \eqref{czalpha}, we are left to show that
\begin{equation}\label{cpeho}-\dfrac{1}{a(\eps)}-\langle\rho^{\eps},(-\Delta_{\sigma}+z)^{-1}\rho^{\eps}\rangle\to \alpha+\frac{\sqrt{z}}{4\pi}-h_z(0).
	\end{equation}
In view of the expression \eqref{green_b} for the Green function of $-\Delta_{\sigma}$, we get
\begin{equation*}
	\begin{split}
\langle\rho^{\eps},(-\Delta_{\sigma}+z)^{-1}\rho^{\eps}\rangle&=\int_{\R^3}\int_{\R^3}\big(\Gamma_z(x,y)+h_z^y(x)\big)\rho^{\eps}(x)\rho^{\eps}(y)dydx\\
&=\frac{\ell}{\eps}-\frac{\sqrt{z}}{4\pi}+h_z(0)+o(1),
	\end{split}
\end{equation*}
which implies \eqref{cpeho} (precisely when the coupling constant $a(\eps)$ has the form \eqref{aeps}), concluding the argument.\\ 
Let us notice that the non-local approximation scheme does not need any spectral conditions.
%Let us notice that, also within non-local approximation scheme, we do not have anymore an explicit notion of resonance as a suitable generalized eigenfunction for the Laplace operator $-\Delta_{\sigma}$ on the bounded domain $\Omega$. {\color{red} $\psi=\mathcal{G}^y_z(x)*\rho$ non è diventato un autovettore di energia zero?}

\medskip

\subsection{Unbounded domains.}
Finally, let us consider the case of an open, \emph{unbounded} domain $\Omega$ different from the whole $\R^3$. Analogously to the case of bounded domains, the approximation of a point interaction via local operators $-\Delta_{\sigma}+V_{\eps}$ can be obtained by the extension-restriction method, as soon as the potential $V$ is compactly supported and the boundary condition $\sigma$ on $\partial\Omega$ guarantees the self-adjointeness of $-\Delta_{\sigma}$, the regularity of its Green function, and a well-definition on Sobolev spaces of the trace/extension maps. Consider for example exterior domains, i.e.~$\Omega$ is the complement of a compact domain with regular boundary. In this case the above requirements are satisfied (see for example \cite{Leis86, Mochizuki17}), so that a statement analogous to Theorem \ref{th:main} holds and its proof applies {\em verbatim}. Apart from the special case of exterior domains, it would be interesting to consider the approximation procedure in the light of the general classification of unbounded domains, i.e.~quasi-conical, quasi-cylindrical and quasi-bounded domains (see \cite{Glazman65, EE18}). Moreover, still under suitable assumptions on $\sigma$, also the approximation scheme via non-local operators can be verified by a direct analysis of the resolvent. An interesting open question, in this context, concerns the possibility (both for local and non-local approximations) of a notion of zero-energy resonance as generalized eigenfunction, exhibiting a non-$L^2$ decay in the spatial directions at infinity.

\section*{Acknowledgments.}
\noindent The first author acknowledges the support of the Next Generation EU - Prin 2022 project "Singular Interactions and Effective Models in Mathematical Physics- 2022CHELC7". The second author acknowledges the support by INdAM-GNAMPA through the project ``Local and nonlocal equations with lower order terms".

\end{document}